\documentclass[runningheads, thm-restate]{llncs}

\usepackage{aligned-overset} 
\usepackage[ruled,vlined]{algorithm2e} 
\usepackage[dvipsnames]{xcolor} 
\usepackage{tikz}
\usepackage{graphicx}
\usepackage{pgfplots}
\usetikzlibrary{decorations}
\usepackage{color}
\usepackage{todonotes}
\SetKwComment{Comment}{// }{} 
\usepackage{multirow} 
\usepackage{thm-restate}

\usepackage{booktabs}

\usepackage[T1]{fontenc}
\usepackage[utf8]{inputenc}

\usepackage{amsmath, amssymb, dsfont}
\usepackage{paralist}
\usepackage{mathrsfs} 

\usepackage{hyperref}
\usepackage{cleveref}
\newcommand{\N}{\mathbb{N}}
\newcommand{\R}{\mathbb{R}}

\newcommand{\alg}{\textsc{Alg}}
\newcommand{\opt}{\textsc{Opt}}
\newcommand{\deliver}{\textsc{deliver\_and\_return}}

\newcommand{\schedule}{\textsc{follow\_schedule}}
\newcommand{\waituntil}{\textsc{wait\_until}}

\newcommand{\wait}{\textsc{Lazy}}
\newcommand{\lazy}{\textsc{Lazy}}
\renewcommand{\o}{O} 
\newcommand{\ti}{t^{(i)}}
\newcommand{\tip}{t^{(i+1)}}

\renewcommand{\epsilon}{\varepsilon}

\spnewtheorem{observation}{Observation}{\bfseries}{\itshape}

\crefformat{lemma}{Lemma~#2#1#3}
\crefformat{theorem}{Theorem~#2#1#3}
\crefformat{proposition}{\mbox{Proposition~#2#1#3}}
\crefformat{remark}{Remark~#2#1#3}
\crefformat{equation}{(#2#1#3)}
\crefformat{corollary}{Corollary~#2#1#3}
\crefformat{example}{Example~#2#1#3}
\crefformat{definition}{Definition~#2#1#3}
\crefformat{observation}{Observation~#2#1#3}
\pgfplotsset{compat=1.14}
\usetikzlibrary{arrows.meta,arrows}
\usetikzlibrary{positioning,arrows}
\usetikzlibrary{patterns}
\usepgfplotslibrary{fillbetween}
\pgfplotsset{every x tick label/.append style={font=\small, yshift=0.0ex}}
\pgfplotsset{every y tick label/.append style={font=\small, yshift=0.0ex}}
\usetikzlibrary{patterns}
\makeatletter
\pgfdeclarepatternformonly[\LineSpace]{my north east lines}{\pgfqpoint{-1pt}{-1pt}}{\pgfqpoint{\LineSpace}{\LineSpace}}{\pgfqpoint{\LineSpace}{\LineSpace}}%
{
	\pgfsetcolor{\tikz@pattern@color}
	\pgfsetlinewidth{0.4pt}
	\pgfpathmoveto{\pgfqpoint{0pt}{0pt}}
	\pgfpathlineto{\pgfqpoint{\LineSpace + 0.1pt}{\LineSpace + 0.1pt}}
	\pgfusepath{stroke}
}
\makeatother
\newdimen\LineSpace
\tikzset{
	line space/.code={\LineSpace=#1},
	line space=10pt
}
\usepgfplotslibrary{fillbetween}
\tikzset{
	schraffiert/.style={pattern=horizontal lines,pattern color=#1},
	schraffiert/.default=black
}
\tikzstyle{densely dashed}=          [dash pattern=on 6pt off 2pt]
\tikzstyle{densely shadow}=          [dash pattern=on 6.5pt off 1.5pt]

\begin{document}
	\title{Tight analysis of the lazy algorithm for open online dial-a-ride\thanks{Supported by DFG grant DI 2041/2.}}
	%
	%
	\author{Júlia Baligács\inst{1}\orcidID{0000-0003-2654-149X} \and
		Yann Disser\inst{1}\orcidID{0000-0002-2085-0454} \and
		Farehe Soheil\inst{1}\orcidID{0000-0002-0504-8834} \and
		David Weckbecker\inst{1}\orcidID{0000-0003-3381-058X}}
	\authorrunning{J.~Baligács, Y.~Disser, F.~Soheil, and D.~Weckbecker}
	%
	\institute{TU Darmstadt, Germany
		\email{\{baligacs|disser|soheil|weckbecker\}@mathematik.tu-darmstadt.de}}
	\maketitle              
	\begin{abstract}
		In the open online dial-a-ride problem, a single server has to deliver transportation requests appearing over time in some metric space, subject to minimizing the completion time. We improve on the best known upper bounds on the competitive ratio on general metric spaces and on the half-line, for both the preemptive and non-preemptive version of the problem. We achieve this by revisiting the algorithm $\textsc{Lazy}$ recently suggested in [WAOA, 2022] and giving an improved and tight analysis. More precisely, we show that it has competitive ratio $2.457$ on general metric spaces and $2.366$ on the half-line. 
This is the first upper bound that beats known lower bounds of 2.5 for schedule-based algorithms as well as the natural $\textsc{Replan}$ algorithm.
		
	\end{abstract}

	\section{Introduction}
	\label{sec:typesetting-summary}
	
	In the open online dial-a-ride problem, we are given a metric space $(M,d)$ and have control of a server that can move at unit speed.
	Over time, \emph{requests} of the form $(a,b;t)$ arrive.
	Here, $a\in M$ is the \emph{starting position} of the request, $b\in M$ is its \emph{destination}, and $t\in\R_{\geq0}$ is the \emph{release time} of the request. We consider the online variant of the problem, meaning that the server does not get to know all requests at time~$0$, but rather at the respective release times.
	Our task is to control the server such that it serves all requests, i.e., we have to move the server to position~$a$, load the request $(a,b;t)$ there after its release time~$t$, and then move to position~$b$ where we unload the request. The objective is to minimize the \emph{completion time}, i.e., the time when all requests are served.

	We assume that the server always starts at time~$0$ in some fixed point, which we call the \emph{origin} $\o\in M$.
	The server has a \emph{capacity} $c\in(\N\cup\{\infty\})$ and  is not allowed to load more than~$c$ requests at the same time.
	Furthermore, we consider the \emph{non-preemptive} version of the problem, that is, the server may not unload a request preemptively at a point that is not the request’s destination.
	In the dial-a-ride problem, a distinction is made between the open and the closed variant.
	In the \emph{closed} dial-a-ride problem, the server has to return to the origin after serving all requests.
	By contrast, in the \emph{open} dial-a-ride problem, the server may finish anywhere in the metric space.
	In this work, we only consider the open variant of the problem.
	By letting $a=b$ for all requests $(a,b;t)$, we obtain the \emph{online travelling salesperson problem (TSP)} as a special case of the dial-a-ride problem.
	
	In this work, we only consider deterministic algorithms for the online dial-a-ride problem. As usual in competitive analysis, we measure the quality of a deterministic algorithm by comparing it to an optimum offline algorithm.
	The measure we apply is the completion time of a solution.
	For a given sequence of requests~$\sigma$ and an algorithm $\alg$, we denote by $\alg(\sigma)$ the completion time of the algorithm for request sequence~$\sigma$.
	Analogously, we denote by $\opt(\sigma)$ the completion time of an optimal offline algorithm.
	For some $\rho\geq1$, we say that an algorithm $\alg$ is \emph{$\rho$-competitive} if, for all request sequences~$\sigma$, we have $\alg(\sigma)\leq\rho\cdot\opt(\sigma)$.
	The \emph{competitive ratio} of $\alg$ is defined as ${\inf\{\rho\geq1\mid\alg\textrm{ is }\rho\textrm{-competitive}\}}$. The \emph{competitive ratio of a problem} is defined as $\inf\{\rho\geq1\mid\textrm{there is some $\rho$-competitive algorithm}\}$
	
	\subsubsection*{Our results.}
	
	We consider the parametrized algorithm $\lazy(\alpha)$ that was presented in \cite{waoa} and prove the following results (see Table~\cref{resultstable}).

	\begin{table}[t]
		\begin{center}
			\begin{tabular}{p{2.5cm}cccc}
				\toprule
				\multirow{2}{*}{metric space} & & \multicolumn{2}{c}{old bounds} & \,\,new bounds\,\,\\\cline{3-4}
				& & lower & upper & upper\\
				\midrule
				\multirow{2}{*}{general}&\,\,non-preemptive\,\,& 2.05 & \textbf{2.618}~\cite{waoa}&\textbf{2.457}~\smaller{(Thm~\ref{mainthm})}\\ & preemptive & 2.04 &2.618 & 2.457\\
				\hline
				\multirow{2}{*}{line}&non-preemptive&\,\,\textbf{2.05}~\cite{BirxDisser/22}\,\,&2.618&2.457\\&preemptive & \textbf{2.04}~\cite{BjeldeDisserHackfeldEtal/20}&\,\,\textbf{2.41}~\cite{BjeldeDisserHackfeldEtal/20}\,\, & ---\\
				\hline
				\multirow{2}{*}{half-line}&non-preemptive & \textbf{1.9}~\cite{Lipmann/03}&2.618 & \,\,\textbf{2.366}~\smaller{(Thm~\ref{half-line_result})}\,\,\\&preemptive & \textbf{1.62}~\cite{Lipmann/03}&2.41&2.366\vspace{-\aboverulesep}\\
				\bottomrule
			\end{tabular}
		\end{center}
		\caption{State of the art of the open online dial-a-ride problem and overview of our results: Bold bounds are original results, other bounds are inherited.}\label{resultstable}\vspace{-2.6em}
	\end{table}
	
	Our main result is an improved general upper bound for the open online dial-a-ride problem.
	
	\begin{theorem}\label{mainthm}
		For $\alpha=\frac{1}{2}+\sqrt{11/12}$, $\lazy(\alpha)$ has a competitive ratio of\linebreak $\alpha+1\thickapprox2.457$ for open online dial-a-ride on general metric spaces for every capacity $c\in\N\cup\{\infty\}$.
	\end{theorem}

    Prior to our work, the best known general upper bound of $\varphi+1\thickapprox 2.618$ on the competitive ratio for the open online dial-a-ride problem was achieved by $\lazy(\varphi)$ and it was shown that $\lazy(\alpha)$ has competitive ratio at least ${\frac{3}{2}+\sqrt{11/12}\thickapprox 2.457}$ for any choice of~$\alpha$, even on the line~\cite{waoa}. 
	This means that we give a conclusive analysis of $\lazy(\alpha)$ by achieving an improved upper bound that tightly matches the previously known lower bound.
	In particular,\linebreak $\alpha=1/2+\sqrt{11/12}$ is the (unique) best possible waiting parameter for $\lazy(\alpha)$, even on the line.
    The best known general lower bound remains~$2.05$~\cite{BirxDisser/22}.
	
	Crucially, our upper bound beats, for the first time, a known lower bound of~$2.5$ for the class of \emph{schedule-based} algorithms~\cite{Birx/20}, i.e., algorithms that divide the execution into subschedules that are never interrupted.
	Historically, all upper bounds, prior to those via $\lazy$, were based on schedule-based algorithms~\cite{BirxDisser/20,BirxDisser/22}.
	Our result means that online algorithms cannot afford to irrevocably commit to serving some subset of requests if they hope to attain the best possible competitive ratio.
	
	Secondly, our upper bound also beats the same lower bound of~$2.5$ for the $\textsc{ReOpt}$ (or $\textsc{Replan}$) algorithm~\cite{Ausiello/01}, which simply reoptimizes its solution whenever new requests appear.
	While this algorithm is very natural and may be the first algorithm studied for the online dial-a-ride problem, it has eluded tight analysis up to this day. 
	So far, it has been a canonical candidate for a best-possible algorithm.
	We finally rule it out. 
	
	In addition to the general bound above, we analyze $\lazy(\alpha)$ for open online dial-a-ride on the half-line, i.e., where $M=\R_{\geq 0}$, and show that, in this metric space, even better bounds on the competitive ratio are possible for different values of~$\alpha$.
	More precisely, we show the following.
	
	\begin{theorem}\label{half-line_result}
		For $\alpha=\frac{1+\sqrt{3}}{2}$, $\lazy(\alpha)$ has a competitive ratio of $\alpha+1\thickapprox2.366$ for open online dial-a-ride on the half-line for every capacity ${c\in\N\cup\{\infty\}}$.
	\end{theorem}
	
	This further improves on the previous best known upper bound of 2.618 \cite{waoa}.
	The best known lower bound is 1.9 \cite{Lipmann/03}.
	
	We go on to show that the bound in \cref{half-line_result} is best-possible for $\wait(\alpha)$ over all parameter choices $\alpha\geq0$.
	
	\begin{theorem}\label{half-line_lower-bound}
		For all $\alpha\geq0$, $\lazy(\alpha)$ has a competitive ratio of at least\linebreak $\alpha+1\thickapprox2.366$ for open online dial-a-ride on the half-line for every capacity $c\in\N\cup\{\infty\}$.
	\end{theorem}
	
	In the preemptive version of the online dial-a-ride problem, the server is allowed to unload requests anywhere and pick them up later again.
	In this version, prior to our work, the best known upper bound on general metric spaces was 2.618~\cite{waoa} and the best known upper bound on the line and the half-line was 2.41~\cite{BjeldeDisserHackfeldEtal/20}.
	Obviously, every non-preemptive algorithm can also be applied in the preemptive setting, however, its competitive ratio may degrade since the optimum might have to use preemption.
	Our algorithm~$\wait$ repeatedly executes optimal solutions for subsets of requests and can be turned preemptive by using preemptive solutions.
    With this change, our analysis of~$\wait$ still carries through in the preemptive case and improves the state of the art for general metric spaces and the half-line, but not the line.
	 The best known lower bound in the preemptive version on general metric spaces is 2.04 \cite{BjeldeDisserHackfeldEtal/20} and the best known lower bound on the half-line is 1.62 \cite{Lipmann/03}.

	\begin{corollary}
		The competitive ratio of the open preemptive online dial-a-ride problem with any capacity $c \in \N \cup \{\infty\}$ is upper bounded by 
		\begin{enumerate}[a)]
			\item $\frac{3}{2}+\sqrt{\frac{11}{12}} \thickapprox 2.457$ and this bound is achieved by $\wait\left(\frac{1}{2}+\sqrt{\frac{11}{12}}\right)$,
			\item $1+\frac{1+\sqrt{3}}{2}\thickapprox 2.366$ on the half-line and this bound is achieved by $\wait(\frac{1+\sqrt{3}}{2})$.
		\end{enumerate}

	\end{corollary}

	\subsubsection*{Related work.}
	
Two of the most natural algorithms for the online dial-a-ride problem are $\textsc{Ignore}$ and $\textsc{Replan}$.
The basic idea of $\textsc{Ignore}$ is to repeatedly follow an optimum schedule over the currently unserved requests and ignoring all requests released during its execution. 
The competitive ratio of this algorithm is known to be exactly 4 \cite{Birx/20,Krumke0}. By contast, the main idea of $\textsc{Replan}$ is to start a new schedule over all unserved requests whenever a new request is released. While this algorithm has turned out to be notoriously difficult to analyze, it is known that its competitive ratio is at least 2.5 \cite{Ausiello/01} and at most 4~\cite{Birx/20}. Several variants of these algorithms have been proposed such as \textsc{SmartStart}~\cite{Krumke0}, \textsc{SmarterStart}~\cite{BirxDisser/22} or \textsc{WaitOrIgnore}~\cite{Lipmann/03}, which lead to improvements on the best known bounds on the competitive ratio of the dial-a-ride problem. In this work, we study the recently suggested algorithm $\lazy$ \cite{waoa}, which is known to achieve a competitive ratio of $\varphi+1$ for the open online dial-a-ride problem, where $\varphi=\smash{\frac{\sqrt{5}+1}{2}}\thickapprox1.618$ denotes the golden ratio.

	For the preemptive version of the open dial-a-ride problem, the best known upper bound on general metric spaces is $\varphi+1\thickapprox 2.618$ \cite{waoa}. 
	Bjelde et al.~\cite{BjeldeDisserHackfeldEtal/20} proved a stronger upper bound of $1+\sqrt{2}\approx 2.41$ for when the metric space is the line. 
	
	In terms of lower bounds, Birx et al.~\cite{BirxDisser/22} were able to prove that every algorithm for the open online dial-a-ride problem has a competitive ratio of at least~2.05, even if the metric space is the line.
	 This separates dial-a-ride from online TSP on the line, where it is known that the competitive ratio is exactly~2.04~\cite{BjeldeDisserHackfeldEtal/20}. 
	For open online dial-a-ride on the half-line, Lipmann \cite{Lipmann/03} established a lower bound of~1.9 for the non-preemptive version and a lower bound of~1.62 for the preemptive version.
	
	
	For the closed variant of the online dial-a-ride problem, the competitive ratio is known to be exactly 2 on general metric spaces \cite{AscheuerKrumkeRambau/00,Ausiello/01,FeuersteinStougie/01} and between~1.76 and~2 on the line \cite{Birx/20,BjeldeDisserHackfeldEtal/20}. 
	On the half-line the best known lower bound is 1.71 \cite{AscheuerKrumkeRambau/00} and the best known upper bound is 2 \cite{AscheuerKrumkeRambau/00,FeuersteinStougie/01}.
	The TSP variant of the closed dial-a-ride problem is tightly analyzed with a competitive ratio of 2 on general metric spaces \cite{AscheuerKrumkeRambau/00,Ausiello/01,FeuersteinStougie/01}, of 1.64 on the line \cite{Ausiello/01,BjeldeDisserHackfeldEtal/20}, and of 1.5 on the half-line \cite{BlomKrumkePaepeStougie/01}.
	
Other variants of the problem have been studied in the literature.
This includes settings where the request sequence has to fulfill some reasonable additional properties \cite{BlomKrumkePaepeStougie/01,HauptmeierKrumkeRabau/00,Krumke/02}, where the server is presented with additional~\cite{AusielloAllulliBonifaciLaura/06}  or less~\cite{Lipmann/04} information, where the server has some additional abilities~\cite{BonifaciStougie/08,JailletWagner/08}, where the server has to handle requests in a given order \cite{HauptmeierKrumkeRambauWirth/01,JailletWagner/08}, or where we consider different objectives than the completion time \cite{AusielloDemangeLauraPaschos/04,BienkowskiKL/21,BienkowskiLiu/19,HauptmeierKrumkeRabau/00,JailletLu/11,JailletLu/14,Krumke/02,KrumkePaepePoensgenEtal/06,KrumkePaepePoensgenStougie/03}. Other examples include the study of randomized algorithms \cite{Krumke0}, or other metric spaces, such as a circle \cite{JawgalMuralidharaSrinivasan/19}. Moreover, it has been studied whether some natural classes of algorithms can have good competitive ratios. For example, \emph{zealous} algorithms always have to move towards an unserved request or the origin \cite{BlomKrumkePaepeStougie/01}. A \emph{schedule-based} algorithm operates in schedules that are not allowed to be interrupted. 
Birx \cite{Birx/20} showed that all such algorithms have a competitive ratio of at least 2.5.
Together with our results, this implies that schedule-based algorithms algorithms cannot be best-possible.

	\section{Algorithm description and notation}
	In this section, we define the algorithm $\wait$ introduced in \cite{waoa}. The rough idea of the algorithm is to wait until several requests are revealed and then start a schedule serving them. Whenever a new request arrives, we check whether we can deliver all currently loaded requests and return to the origin in a reasonable time. If this is possible, we do so and begin a new schedule including the new requests starting from the origin. If this is not possible, we keep following the current schedule and consider the new request later.
	
	More formally, a \emph{schedule} is a sequence of actions specifying the server's  behaviour, including its movement and where requests are loaded or unloaded. By $\opt[t]$, we denote an optimal schedule beginning in $\o$ at time 0 and serving all requests that are released not later than time $t$. By $\opt(t)$, we denote its completion time.
	Given a set of requests $R$ and some point $x \in M$, we denote by $S(R,x)$ a shortest schedule serving all requests in $R$ beginning from point $x$ at some time after all requests in $R$ are released. In other words, we can ignore the release times of the requests when computing $S(R,x)$. As waiting is not beneficial for the server if there are no release times, the \emph{length of the schedule}, i.e., the distance the server travels, is the same as the time needed to complete it and we denote this by $|S(R,x)|$.
	
	Now that we have established the notation needed, we can describe the algorithm (cf. \hyperref[lazy]{Algorithm 1}). 
	By $t$, we denote the current time. By $p_t$, we denote the position of the server at time $t$, and by $R_t$, we denote the set of requests that have been released but not served until time $t$. The variable~$i$ is a counter over the schedules started by the algorithm. The waiting parameter $\alpha\geq 1$ specifies how long we wait before starting a schedule. The algorithm uses the following commands: $\deliver$ orders the server to finish serving all currently loaded requests and return to $\o$ in the fastest possible way, $\waituntil(t)$ orders the server to remain at its current location until time $t$, and $\schedule(S)$ orders the server to execute the actions defined by schedule $S$.
	When any of these commands is invoked, the server aborts what it is doing and executes the new command. Whenever the server has completed a command, we say that it becomes \emph{idle}.
	
	\begin{algorithm}\label{lazy}
		\caption{$\wait(\alpha)$}
		initialize: $i \gets 0$
		
		\vspace{.1cm}
		\hrule
		\vspace{.1cm}
		
		\emph{upon receiving a request:}\\
		\If{server can serve all loaded requests and return to $\o$ until time $\alpha \cdot \opt (t)$}{
			\textbf{execute} $\deliver$ 
		}
		
		\hrule
		\vspace{.1cm}
		
		\emph{upon becoming idle:}\\
		\uIf{$t<\alpha\cdot\opt(t)$}{
			\textbf{execute} $\waituntil(\alpha \cdot \opt (t))$} 
		{\ElseIf{$R_t\neq\emptyset$}{
				$i\gets i+1$, $R^{(i)}\gets R_t$, $t^{(i)}\gets t$, $p^{(i)}\gets p_t$ \\
				$S^{(i)}\gets S(R^{(i)},p^{(i)})$\\
				\textbf{execute} $\schedule ({S^{(i)})}$
		}} 		
	\end{algorithm}
	We make a few comments for illustration of the algorithm.
	If the server returns to the origin upon receiving a request, we say that the schedule it was currently following is \emph{interrupted}.
	Observe that, due to interruption, the sets~$R^{(i)}$ are not necessarily disjoint.
	Also, observe that $p^{(1)}=\o$, and if schedule $S^{(i)}$ was interrupted, we have $p^{(i+1)}=\o$ and $t^{(i+1)}=\alpha\cdot\opt(t^{(i+1)})$.
	If $S^{(i)}$ was not interrupted, $p^{(i+1)}$ is the ending position of $S^{(i)}$.
	
	The following observations were already noted in \cite{waoa} and follow directly from the definitions above and the fact that requests in $R^{(i)} \setminus R^{(i-1)}$ were released after time $t^{(i-1)}$.
	
	\begin{restatable}[\cite{waoa}]{observation}{DefObs}\label{defobs}
		For every request sequence, the following hold.
		\begin{enumerate}[a)] 
			\item For every $i>1$, $\opt(\ti)\geq t^{(i-1)}\geq\alpha\cdot\opt(t^{(i-1)})$.
			\item For every $x,y \in M$ and every subset of requests $R$, we have
			$|S(R,x)|\leq d(x,y)+|S(R,y)|$.
			\item Let $i>1$ and assume that $S^{(i-1)}$ was not interrupted. Let $a$ be the starting position of the request in $R^{(i)}$ that is picked up first by $\opt(t^{(i)})$. Then, 
			\begin{equation*}
				\opt(t^{(i)})\geq t^{(i-1)}+|S(R^{(i)},a)|\geq \alpha\cdot\opt(t^{(i-1)})+|S(R^{(i)},a)|.
			\end{equation*}
		\end{enumerate}
	\end{restatable}

	\section{Factor-revealing approach}\label{sec:factorRevealing}
The results in this paper were informed by a factor-revealing technique, inspired by a similar approach of Bienkowski et al.~\cite{BienkowskiKL/21}, to analyze a specific algorithm~\alg, in our case $\lazy$. 
The technique is based on a formulation of the adversary problem, i.e., the problem of finding an instance that maximizes the competitive ratio, as an optimization problem of the form
	\begin{align}
	\max\biggl\{\frac{\alg(x)}{\opt(x)}\Bigm| x\textrm{ describes a dial-a-ride instance}\biggr\}.\label{eq:opt_problem}
	\end{align}
	An optimum solution to this problem immediately yields the competitive ratio of~\alg.
	Of course, we cannot hope to solve this optimization problem or even describe it with a finite number of variables.	
	The factor-revealing approach consists in relaxing~\cref{eq:opt_problem} to a practically solvable problem over a finite number of variables.
	
	The key is to select a set of variables that captures the structure of the problem well enough to allow for meaningful bounds.
	In our case, we can, for example, introduce variables for the starting position and duration of the second-to-last as well as for the last schedule.
	We then need to relate those variables via constraints that ensure that an optimum solution to the relaxed problem actually has a realization as a dial-a-ride instance.
	For example, we might add the constraint that the distance between the starting positions of the last two schedules is upper bounded by the duration of the second-to-last schedule.
	
	The power of the factor-revealing approach is that it allows to follow an iterative process for deriving structurally crucial inequalities:
	When solving the relaxed optimization problem, we generally have to expect an optimum solution that is not realizable and overestimates the competitive ratio.
	We can then focus our efforts on understanding why the corresponding variable assignment cannot be realized by a dial-a-ride instance. Then, we can introduce additional variables and constraints to exclude such solutions.
	In this way, the unrealizable solutions inform our analysis in the sense that we obtain bounds on the competitive ratio that can be proven analytically by only using the current set of variables and inequalities.  
	Once we obtain a realizable lower bound, we thus have found the exact competitive ratio of the algorithm under investigation.
	
	In order to practically solve the relaxed optimization problems, we limit ourselves to linear programs (LPs).
	Note that the objective of~\cref{eq:opt_problem} is linear if we normalize to $\opt(x)=1$.
	We can do this, since the competitive ratio is invariant with respect to rescaling the metric space and release times of requests.
	Another advantage of using linear programs is that we immediately obtain a formal proof of the optimum solution from an optimum solution to the LP dual.
	Of course, the correctness of the involved inequalities still needs to be established.
	
	In the remainder of this paper, we present a purely analytic proof of our results.
	Many of the inequalities we derive in lemmas were informed by a factor-revealing approach via a linear program with a small number of binary variables.
	This means that we additionally need to branch on all binary variables in order to obtain a formal proof via LP duality.
	We refer to \Cref{sec:fr_hl} for more details of the binary program that informed our results for the half-line.

	\section{Analysis on general metric spaces}

	This section is concerned with the proof of \cref{mainthm}. For the remainder of this section, let $(r_1,\dots,r_n)$ be some fixed request sequence. Let $k$ be the number of schedules started by $\wait(\alpha)$, and let $S^{(i)}$, $t^{(i)}$, $p^{(i)}$, $R^{(i)}$ ($1 \leq i \leq k$) be defined as in the algorithm. Note that we slightly abuse notation here because $k$, $S^{(i)}$, $t^{(i)}$, $p^{(i)}$, and $R^{(i)}$ depend on $\alpha$. 
	As it will always be clear from the context what~$\alpha$ is, we allow this implicit dependency in the notation. 
	
	\newcommand{\rflaz}{r_{f,\wait}^{(i)}}
	\newcommand{\aflaz}{a_{f,\wait}^{(i)}}
	\newcommand{\bflaz}{b_{f,\wait}^{(i)}}
	\newcommand{\tflaz}{t_{f,\wait}^{(i)}}
	\newcommand{\rllaz}{r_{l,\wait}^{(i)}}
	\newcommand{\allaz}{a_{l,\wait}^{(i)}}
	\newcommand{\bllaz}{b_{l,\wait}^{(i)}}
	\newcommand{\tllaz}{t_{l,\wait}^{(i)}}
	\newcommand{\rfopt}{r_{f,\opt}^{(i)}}
	\newcommand{\afopt}{a_{f,\opt}^{(i)}}
	\newcommand{\bfopt}{b_{f,\opt}^{(i)}}
	\newcommand{\tfopt}{t_{f,\opt}^{(i)}}
	\newcommand{\rlopt}{r_{l,\opt}^{(i)}}
	\newcommand{\alopt}{a_{l,\opt}^{(i)}}
	\newcommand{\blopt}{b_{l,\opt}^{(i)}}
	\newcommand{\tlopt}{t_{l,\opt}^{(i)}}
	\newcommand{\aafopt}{a_{f,\opt}^{(i+1)}}
	
	As it will be crucial for the proof in which order $\opt$ and $\wait$ serve requests, we introduce the following notation. Let
	\begin{itemize}
		\item $\rfopt=(\afopt,\bfopt;\tfopt)$ be the first request in $R^{(i)}$ picked up by $\opt[t^{(i)}]$,
		\item $\rlopt=(\alopt,\blopt;\tlopt)$ be the last request in $R^{(i)}$ delivered by\linebreak $\opt[\tip]$,
		\item $\rflaz=(\aflaz, \bflaz; \tflaz)$ be the first request in $R^{(i)}$ picked up by $\wait (\alpha)$,
		\item $\rllaz=(\allaz, p^{(i+1)};\tllaz)$ be the last request in $R^{(i)}$ delivered by $\wait (\alpha)$.
	\end{itemize}

	\begin{definition}
		We say that the $i$-th schedule is \emph{$\alpha$-good} if 
		\begin{enumerate}[a)]
			\item $|S^{(i)}| \leq \opt(t^{(i)})$ and
			\item $t^{(i)} + |S^{(i)}| \leq (1+\alpha) \cdot\opt(t^{(i)})$.
		\end{enumerate}
	\end{definition}
	
	In this section, we prove by induction on $i$ that, for $\alpha \geq \frac{1}{2}+\sqrt{11/12}$, every schedule is $\alpha$-good. Note that this immediately implies \cref{mainthm}. 
	
	As our work builds on \cite{waoa}, the first few steps of our proof are the same as in~\cite{waoa}. For better understandability and reading flow, we repeat the proofs of some important but simple steps and mark the results with appropriate citations. 
	The results starting with \cref{wrongorder} are new and improve on the analysis in~\cite{waoa}.
	
	We begin with proving the base case.
	
	\begin{observation}[Base case, \cite{waoa}]\label{basecase}
		For every $\alpha \geq 1$, the first schedule is $\alpha$-good.
	\end{observation}
	\begin{proof}
		Recall that $S^{(1)}$ begins in $\o$ and is the shortest tour serving all requests in $R^{(1)}$.
		$\opt [t^{(1)}]$ begins in $\o$ and serves all requests in $R^{(1)}$, too, which yields $|S^{(1)}|\leq \opt(t^{(1)})$.
		The fact that we have $t^{(1)}=\alpha \cdot \opt (t^{(1)})$ implies\linebreak
		$t^{(1)}+|S^{(1)}| \leq (1+\alpha) \cdot \opt (t^{(1)})$.
		\qed
	\end{proof}
	
	Next, we observe briefly that the induction step is not too difficult when the last schedule was interrupted.
	
	\begin{observation}[Interruption case, \cite{waoa}]\label{interrupted}
		Let $\alpha \geq 1$. Assume that schedule~$S^{(i)}$ was interrupted. Then, $S^{(i+1)}$ is $\alpha$-good.
	\end{observation}
	\begin{proof}
		If schedule $S^{(i)}$ was interrupted, we have $p^{(i+1)}=\o$ and\linebreak $t^{(i+1)}=\alpha \cdot \opt (t^{(i+1)})$.
		Therefore, $|S^{(i+1)}|=|S(R^{(i+1)},\o)|\leq \opt (t^{(i+1)})$ and $t^{(i+1)}+|S^{(i+1)}| \leq (1+\alpha) \cdot \opt (t^{(i+1)})$.
		\qed
	\end{proof}
	
	For this reason, we will assume in many of the following statements that the schedule $S^{(i)}$ was not interrupted.
	
	By careful observation of the proof in \cite{waoa}, one can see that the following fact already holds for smaller $\alpha$. For convenience, we repeat the proof of the following Lemma with an adapted value of~$\alpha$.
	
	\begin{lemma}[\cite{waoa}]\label{waoa}
		Let  ${\alpha\geq \frac{1+\sqrt{17}}{4}\thickapprox 1.281}$ and $i\in \{1, \dots, k-1\}$. If $S^{(i)}$ is $\alpha$-good, then
		$|S^{(i+1)}| \leq \opt (t^{(i+1)})$.
	\end{lemma}
	\begin{proof}	
		First, observe that if $S^{(i)}$ was interrupted, we have $p^{(i+1)}=\o$. Note that $\opt(\tip)$ begins in $\o$ and serves all requests in $R^{(i+1)}$ so that we have
		\begin{equation*}
		|S^{(i+1)}|=|S(R^{(i+1)}, \o)|\leq\opt(\tip).
		\end{equation*}
		Therefore, assume from now on that $S^{(i)}$ was not interrupted. Also, if $\opt[\tip]$ serves $\rllaz$ at $p^{(i+1)}$ before collecting any request from $R^{(i+1)}$, we trivially have
		\begin{equation*}
		|S^{(i+1)}|=|S(R^{(i+1)},p^{(i+1)}|\leq \opt(\tip).
		\end{equation*}
		Therefore, assume additionally that $\opt[\tip]$ collects $r_{f,\opt}^{(i+1)}$ before serving $\rllaz$. Next, we prove the following assertion.\\
		\emph{Claim: In the setting described above, we have}\begin{equation}\label{eq:claim}
		d(\aafopt,p^{(i+1)}) \leq \left(1+\frac{2}{\alpha}-\alpha\right)\opt(t^{(i)}).
		\end{equation}
		To prove the claim, note that $r_{f,\opt}^{(i+1)}$ is released not earlier than $\alpha\cdot\opt(t^{(i)})$. Since we assume that $\opt(\tip)$ collects $r_{f,\opt}^{(i+1)}$ before serving $\rllaz$ at $p^{(i+1)}$, we obtain
		\begin{equation}\label{eq:wrongorderopt}
		\opt(\tip)\geq\alpha\cdot\opt(\ti)+d(\aafopt,p^{(i+1)}).
		\end{equation}
		Upon the arrival of the last request in $R^{(i)}$, we have $\opt(t)=\opt(\tip)$ and the server can finish its current schedule and return to the origin in time $t^{(i)}+|S^{(i)}|+d(p^{(i+1)},\o)$. As we assume that $S^{(i)}$ was not interrupted, this yields
		\begin{equation}\label{eq:noreturn}
		t^{(i)}+|S^{(i)}|+d(p^{(i+1)},\o)>\alpha\cdot\opt(\tip).
		\end{equation}
		Combined, we obtain that
		\begin{align*}
		d(\aafopt,p^{(i+1)})
		\overset{\text{\cref{eq:wrongorderopt}}}&{\leq} \opt(\tip)-\alpha\cdot\opt(\ti)\\
		\overset{\cref{eq:noreturn}}&{\leq}
		\frac{1}{\alpha}\cdot \left( \ti + |S^{(i)}| + d(p^{(i+1)},\o) \right) - \alpha\cdot\opt(\ti)\\
		\overset{\text{$S^{(i)} \alpha$-good}}&{\leq} \frac{1}{\alpha}\cdot \left( (1+\alpha)\cdot \opt(\ti) + d(p^{(i+1)},\o) \right) - \alpha\cdot\opt(\ti)\\
		&\leq \left( 1+\frac{2}{\alpha} - \alpha \right) \opt(\ti),
		\end{align*}
		where we have used in the last inequality that $d(p^{(i+1)}, \o) \leq \opt(\ti)$ because $\opt(\ti)$ begins in $\o$ and has to serve $\rllaz$ at $p^{(i+1)}$. This completes the proof of the claim.

		Now, we turn back to proving \cref{waoa}. We obtain
		\begin{align*}
		|S^{(i+1)}|&\leq d(p^{(i+1)},\aafopt)+|S(R^{(i+1)},\aafopt)|\\
		\overset{\text{Obs \ref{defobs}c)}}&{\leq} d(p^{(i+1)},\aafopt)+\opt(\tip)-\alpha\cdot\opt(\ti)\\
		\overset{\eqref{eq:claim}}&{\leq} \left(1+\frac{2}{\alpha}-2\alpha\right)\opt(t^{(i)})+\opt(\tip)\\
		&\leq \opt(\tip),
		\end{align*}
		where the last inequality follows from the fact that $1+\frac{2}{\alpha}-2\alpha\leq 0$ if and only if $\alpha\geq \frac{1+\sqrt{17}}{4}\thickapprox 1.2808$.
		\qed
	\end{proof}
	
	Recall that the goal of this section is to prove that every schedule is $\alpha$-good. So far, we have proven the base case (cf. \cref{basecase}) and\linebreak $|S^{(i+1)}|\leq \opt(t^{(i+1)})$ (\cref{waoa}) in the induction step. 
	It remains to show that $t^{(i+1)} + |S^{(i+1)}| \leq (1+\alpha) \cdot \opt (t^{(i+1)})$ assuming $S^{(1)}, \dots, S^{(i)}$ are $\alpha$-good. In \cref{interrupted}, we have already seen that this holds if $S^{(i)}$ was interrupted. To show that the induction step also holds if $S^{(i)}$ was not interrupted, we distinguish several cases for the order in which $\opt$ serves the requests. We begin with the case that $\opt[t^{(i+1)}]$ picks up some request in $R^{(i+1)}$ before serving $\rllaz$, i.e., that $\opt[t^{(i+1)}]$ does not follow the order of the $S^{(i)}$.

	%
	
	\begin{lemma}\label{wrongorder}
		Let $\alpha \geq 1$. Assume that $S^{(i)}$ is $\alpha$-good and was not interrupted, and that $\opt[t^{(i+1)}]$ picks up $r_{f,\opt}^{(i+1)}$ before serving $r_{l,\wait}^{(i)}$. Then, $t^{(i+1)}+|S^{(i+1)}|\leq (1+\alpha) \cdot \opt (t^{(i+1)})$.
	\end{lemma}
	\begin{proof}
		Using the order in which $\opt$ handles the requests, we obtain the following.
		After picking up $r_{f,\opt}^{(i+1)}$ at $\aafopt$ after time $t^{(i)}$, $\opt[\tip]$ has to serve $\rllaz$ at $p^{(i+1)}$ so that 
		\begin{equation}\label{eq:abefore}
		\opt(\tip)\geq t^{(i)}+d(p^{(i+1)},\aafopt).
		\end{equation} 
		After finishing schedule~$S^{(i)}$, the server either waits until time $\alpha\cdot\opt(t^{(i+1)})$ or immediately starts the next schedule, i.e., we have
		\begin{equation*}\label{eq:ti1}
		t^{(i+1)} = \max\{\alpha\cdot\opt(t^{(i+1)}),t^{(i)}+|S^{(i)}|\}.
		\end{equation*}
		If $t^{(i+1)}=\alpha\cdot\opt(t^{(i+1)})$, the assertion follows immediately from \cref{waoa}. Thus, assume $t^{(i+1)}=t^{(i)}+|S^{(i)}|$. This yields
		\begin{align*}
		t^{(i+1)}+|S^{(i+1)}|
		\overset{\text{Obs \ref{defobs}b)}}&{\leq} t^{(i)}+|S^{(i)}|+d(p^{(i+1)},\aafopt)+|S(R^{(i+1)},\aafopt)| \\
		\overset{\text{$S^{(i)} \alpha$-good}}&{\leq} 
		(1+\alpha) \cdot \opt (t^{(i)}) +d(p^{(i+1)},\aafopt)\\
		&\hspace*{1cm}+|S(R^{(i+1)},\aafopt)|\\
		\overset{\text{Obs \ref{defobs}a)}}&{\leq}
		\frac{1+\alpha}{\alpha}t^{(i)} +d(p^{(i+1)},\aafopt)+|S(R^{(i+1)},\aafopt)|\\
		\overset{\text{Obs \ref{defobs}c)}}&{\leq}
		\frac{1}{\alpha}t^{(i)} +d(p^{(i+1)},\aafopt)+\opt(\tip)\\
		\overset{\text{\cref{eq:abefore}}}&{\leq}
		\frac{1}{\alpha}\left(\opt(\tip)-d(p^{(i+1)},\aafopt)\right)+d(p^{(i+1)},\aafopt)\\
		&\hspace{1cm}+\opt(\tip)\\
		&=\left(1+\frac{1}{\alpha} \right) \opt(\tip) + \left(1-\frac{1}{\alpha} \right)d(p^{(i+1)},\aafopt)\\
		&\leq 2 \cdot\opt(\tip),
		\end{align*}	
		where we have used in the last inequality that $	d(p^{(i+1)},\aafopt)\leq \opt(\tip)$ as $\opt[\tip]$ has to visit both points.
		\qed
	\end{proof}
	
	Next, we consider the case where $\opt$ handles $\rllaz$ and $r_{f,\opt}^{(i+1)}$ in the same order as $\wait$.
	
	
	\begin{lemma}\label{rightorder1}
		Let $\alpha\geq 1$. Assume that schedules $S^{(1)}, \dots, S^{(i)}$ are $\alpha$-good, $S^{(i)}$ was not interrupted, and $\opt[t^{(i+1)}]$ serves $\rllaz$ before collecting $r_{f,\opt}^{(i+1)}$. If we have $d(p^{(i+1)},\aafopt)+\opt(t^{(i)})\leq\alpha\cdot\opt(t^{(i+1)})$, then $t^{(i+1)}+|S^{(i+1)}|\leq (1+\alpha) \cdot \opt (t^{(i+1)})$.
	\end{lemma}
	\begin{proof}
		Similarly as in the proof of \cref{wrongorder}, we can assume
		\begin{equation}\label{eq:noWaiting}
		t^{(i+1)}=t^{(i)}+|S^{(i)}|.
		\end{equation}
		We have
		\begin{align*}
		t^{(i+1)}+|S^{(i+1)}| 
		\overset{\text{Obs \ref{defobs}b)}}&{\leq} t^{(i)}+|S^{(i)}|+d(p^{(i+1)},\aafopt)+|S(R^{(i+1)},\aafopt)| \\
		\overset{\text{$S^{(i)} \alpha$-good}}&{\leq} 
		(1+\alpha) \cdot \opt (t^{(i)}) +d(p^{(i+1)},\aafopt)\\
		&\hspace*{1cm}+|S(R^{(i+1)},\aafopt)|\\
		\overset{\text{Obs \ref{defobs}c)}}&{\leq} (1+\alpha)\cdot\opt(t^{(i)})+d(p^{(i+1)},\aafopt)\\
		& \hspace{1cm}+\opt(t^{(i+1)}) - \alpha \cdot \opt (t^{(i)}) \\
		&= \opt(t^{(i+1)})  + d(p^{(i+1)},\aafopt)+\opt(t^{(i)})\\
		&\leq (1+\alpha) \cdot \opt(\tip ),
		\end{align*}
		where the last inequality follows from the assumption that\linebreak $d(p^{(i+1)},\aafopt)+\opt(t^{(i)})\leq\alpha\cdot\opt(t^{(i+1)})$.
		\qed
	\end{proof}
	
	Now that we have proven the case described in \cref{rightorder1}, we will assume in the following that 
	\begin{equation}\label{eq:dlarge}
		d(p^{(i+1)},\aafopt)>\alpha\cdot\opt(t^{(i+1)})-\opt(t^{(i)}).
	\end{equation}
	The following lemma states that, in this case, the $( i-1)$-th schedule (if it exists) was interrupted, i.e., the $i$-th schedule starts in the origin at time $\alpha\cdot\opt(t^{(i)})$. 

	\begin{lemma}\label{rightorder2}
		Let $\alpha\geq \frac{1+\sqrt{3}}{2}\thickapprox 1.366$. Assume that the $i$-th schedule is $\alpha$-good and was not interrupted, and $\opt[t^{(i+1)}]$ serves $\rllaz$ before collecting $r_{f,\opt}^{(i+1)}$. If~\cref{eq:dlarge} holds, then $p^{(i)}=\o$ and $t^{(i)}=\alpha\cdot\opt(t^{(i)})$.
	\end{lemma}
	\begin{proof}
		If $i=1$, we obviously have $p^{(i)}=\o$ and $t^{(i)}=\alpha\cdot\opt(t^{(i)})$.
		Thus, assume that $i\geq2$.
		If \mbox{$\rllaz\in (R^{(i-1)}\cap R^{(i)})$}, schedule $S^{(i-1)}$ was interrupted and, thus,  the statement holds.
		Otherwise, request $\rllaz$ is released while schedule $S^{(i-1)}$ is running, i.e.,  $\tllaz\geq t^{(i-1)}\geq \alpha \cdot \opt(t^{(i-1)})$. Combining this with the assumption that $\opt[t^{(i+1)}]$ serves $\rllaz$ before collecting $r_{f,\opt}^{(i+1)}$, we obtain
		\begin{equation}\label{eq:t_x-upperBound}
		\opt(t^{(i+1)})\geq \tllaz+d(p^{(i+1)},\aafopt) \geq \alpha \cdot \opt(t^{(i-1)})+d(p^{(i+1)},\aafopt).
		\end{equation}
		Rearranging yields
		\begin{align*}
		\alpha\cdot\opt(t^{(i-1)}) \overset{\eqref{eq:t_x-upperBound}}&{\leq} \opt(t^{(i+1)})-d(p^{(i+1)},\aafopt) \\
		\overset{\eqref{eq:dlarge}}&{<} \opt(t^{(i)})-(\alpha-1)\opt(t^{(i+1)}) \\
		\overset{\text{Obs \ref{defobs}a)}}&{\leq} (1+\alpha-\alpha^2)\opt(t^{(i)}),
		\end{align*}
		which is equivalent to
		\begin{equation}\label{eq:estimate_t_i-1}
		\opt(t^{(i-1)}) < \bigl(1+\frac{1}{\alpha}-\alpha\bigr)\opt(t^{(i)}).
		\end{equation}
		By the assumption that $S^{(i-1)}$ is $\alpha$-good, the server finishes schedule $S^{(i-1)}$ not later than time $(\alpha+1)\cdot\opt(t^{(i-1)})$.
		Thus, at the time where request~$\rllaz$ is released, the server can serve all loaded requests and return to the origin by time
		\begin{equation*}
		\max\{(\alpha+1)\opt(t^{(i-1)}),\tllaz\}+\opt(t^{(i-1)}).
		\end{equation*}
		We have
		\begin{align*}
		(\alpha+2)\cdot\opt(t^{(i-1)})\overset{\eqref{eq:estimate_t_i-1}}&{<}(\alpha+2)\left(1+\frac{1}{\alpha}-\alpha\right)\opt(t^{(i)})\\
		& = \left(3+\frac{2}{\alpha}-\alpha^2-\alpha\right)\opt(t^{(i)})\\
		& \leq \alpha\cdot\opt(t^{(i)}),
		\end{align*}
		where the last inequality holds for $\alpha\geq 1.343$.
		Furthermore, as $\rllaz\in R^{(i)}$, it holds that
		\begin{align*}
		\tllaz+\opt(t^{(i-1)}) &\leq \opt(t^{(i)})+\opt(t^{(i-1)})\\
		& \overset{\eqref{eq:estimate_t_i-1}}{\leq} \left(2+\frac{1}{\alpha}-\alpha\right)\opt(t^{(i)})\leq\alpha\cdot\opt(t^{(i)}),
		\end{align*}
		where the last inequality holds for $\alpha \geq \frac{1+\sqrt{3}}{2}\thickapprox 1.366$.
		This implies that the server can return to the origin by time $\alpha\cdot\opt(t^{(i)})$, i.e., we have $p^{(i)}=\o$.
		\qed
	\end{proof}

   We now come to the technically most involved case.
	
	\begin{restatable}{lemma}{RightOrderThree}\label{rightorder3}
		Let $\alpha\geq \frac{1}{2}+\sqrt{11/12} \thickapprox 1.457$. Assume that the $i$-th schedule is $\alpha$-good and was not interrupted, and $\opt[t^{(i+1)}]$ serves $\rllaz$ before collecting $r_{f,\opt}^{(i+1)}$. If \eqref{eq:dlarge} holds, then $t^{(i+1)}+|S^{(i+1)}|\leq (1+\alpha) \cdot \opt (t^{(i+1)})$.
	\end{restatable}
	
	\begin{proof}
		We begin by proving the following assertion.\\
		\emph{Claim: $\opt[t^{(i+1)}]$ serves all requests in $R^{(i)}$ before picking up~$r_{f,\opt}^{(i+1)}$ in~$\aafopt$.}\\	
		To prove the claim, assume otherwise, i.e., that $\opt[t^{(i+1)}]$ serves $\rlopt$  after collecting $r_{f,\opt}^{(i+1)}$.
		The request~$r_{f,\opt}^{(i+1)}$ is released after schedule~$S^{(i)}$ is started, i.e., after time $\alpha\cdot\opt(t^{(i)})$. Thus,
		\begin{align}\label{eq:case221_opt_estimate1}
			\opt(t^{(i+1)}) &\geq \alpha\cdot\opt(t^{(i)})+d(\aafopt,\blopt) \nonumber \\
			\overset{\triangle\text{-ineq}}&{\geq}\!\!\!\!\! \alpha\cdot\opt(t^{(i)})\!+\!d(\aafopt,p^{(i+1)})\!-\!d(\blopt,\o)\!-\!d(\o,p^{(i+1)}).
		\end{align}
		Since $S^{(i)}$ starts in~$\o$, ends in $p^{(i+1)}$ and serves $\rlopt$, we obtain
		\begin{equation}\label{eq:d(0,x_ell)-estimate}
			d(\o,\blopt)+d(\blopt,\o) \leq |S^{(i)}|+d(p^{(i+1)},\o) \!\!\overset{\text{Lem \ref{waoa}}}{\leq}\!\! \opt(t^{(i)})+d(p^{(i+1)},\o).
		\end{equation}
		Furthermore, because $\opt[t^{(i+1)}]$ serves~$\rllaz$ at~$p^{(i+1)}$ before picking up~$r_{f,\opt}^{(i+1)}$ at~$\aafopt$, we have
		\begin{align}
			\opt(t^{(i+1)})&\geq d(\o,p^{(i+1)})+d(p^{(i+1)},\aafopt)\nonumber\\
			\overset{\eqref{eq:dlarge}}&{>} d(\o,p^{(i+1)}) + \alpha\cdot\opt(t^{(i+1)})-\opt(t^{(i)}).\label{eq:d(0,x_i+1)-estimate}
		\end{align}
		Combining all of the above yields
		\begin{align*}
			\opt(t^{(i+1)})
			\overset{\eqref{eq:case221_opt_estimate1},\eqref{eq:d(0,x_ell)-estimate}}&{\geq} \alpha\cdot\opt(t^{(i)})+d(\aafopt,p^{(i+1)})\\
			&-\frac{\opt(t^{(i)})+d(p^{(i+1)},\o)}{2} -d(\o,p^{(i+1)})\\[-.1cm]
			\overset{\eqref{eq:d(0,x_i+1)-estimate}}&{>} \alpha\cdot\opt(t^{(i)})+d(\aafopt,p^{(i+1)})-\frac{\opt(t^{(i)})}{2}\\
			&\hspace*{1cm}-\frac{3}{2}\left(\opt(t^{(i)})-(\alpha-1)\opt(t^{(i+1)})\right) \\[-.1cm]
			&= \left(\frac{3}{2}\alpha-\frac{3}{2}\right)\opt(t^{(i+1)})-(2-\alpha)\opt(t^{(i)})+d(\aafopt,p^{(i+1)}) \\[-.1cm]
			\overset{\eqref{eq:dlarge}}&{>}  \left(\frac{3}{2}\alpha-\frac{3}{2}\right)\opt(t^{(i+1)})-(2-\alpha)\opt(t^{(i)})\\
			&\hspace*{1cm} +\alpha\cdot\opt(t^{(i+1)})-\opt(t^{(i)}) \\[-.1cm]
			&= \left(\frac{5}{2}\alpha-\frac{3}{2}\right)\opt(t^{(i+1)})-(3-\alpha)\opt(t^{(i)}) \\[-.1cm]
			\overset{\text{Obs \ref{defobs}a)}}&{\geq} \left(\frac{5}{2}\alpha-\frac{3}{2}\right)\opt(t^{(i+1)})-\left(\frac{3}{\alpha}-1\right)\opt(t^{(i+1)})\\[-.1cm]
			&=\left(	\frac{5}{2}\alpha-\frac{1}{2}-\frac{3}{\alpha}\right)\opt(t^{(i+1)}) \\[-.1cm]
			&\geq  \opt(t^{(i+1)})
		\end{align*}
		where the last inequality holds if and only if $\alpha\geq \frac{1}{10}\cdot(3+\sqrt{129})\approx 1.436$. As this is a contradiction, we have that $\opt[t^{(i+1)}]$ serves all requests in $R^{(i)}$ before picking up~$r_{f,\opt}^{(i+1)}$ in~$\aafopt$. This completes the proof of the claim. \\
		
		Now that we have established the claim, we turn back to the proof of \cref{rightorder3}.
		Let~$T\geq 0$ denote the time it takes $\opt[t^{(i+1)}]$ until it has served $\rlopt$, i.e., all requests from $R^{(i)}$. First, observe that
		\begin{equation}\label{eq:Tlower}
			T \geq \opt(t^{(i)}).
		\end{equation}
		By the claim, we have
		\begin{equation}\label{eq:case222_opt_i+1_estimate}
			\opt(t^{(i+1)}) \geq T+d(\blopt,\aafopt)+|S(R^{(i+1)},\aafopt)|.
		\end{equation}
		The algorithm $\wait(\alpha)$ finishes $R^{(i+1)}$ by time
		\begin{align}
			t^{(i+1)}\!+\!S^{(i+1)} \overset{\text{Lem \ref{rightorder2}}}&{=}
			\alpha\cdot\opt(t^{(i)})+|S^{(i)}|+|S^{(i+1)}| \nonumber \\
			&\leq \alpha\cdot\opt(t^{(i)})+|S^{(i)}|+d(p^{(i+1)},\aafopt)+|S(R^{(i+1)},\aafopt)| \nonumber \\
			&\leq \alpha\cdot\opt(t^{(i)})+|S^{(i)}|+d(p^{(i+1)},\blopt)+d(\blopt,\aafopt) \nonumber \\
			& \hspace{1cm}+|S(R^{(i+1)},\aafopt)| \nonumber \\[-.2cm]
			\overset{\eqref{eq:case222_opt_i+1_estimate}}&{\leq} \alpha\cdot\opt(t^{(i)}) \!+\! |S^{(i)}| \!+\! d(p^{(i+1)},\blopt) \!+\! \opt(t^{(i+1)}) \!-\! T.\label{eq:case222_alg_estimate1}
		\end{align}
		As $S^{(i)}$ visits $\blopt$ before $p^{(i+1)}$ and $\opt[t^{(i+1)}]$ visits $p^{(i+1)}$ before $\blopt$,
		\begin{align}
			|S^{(i)}|+T &\geq \left(d(\o,\blopt)+d(\blopt,p^{(i+1)})\right) \nonumber\\
			&\hspace*{1cm}+\left(d(\o,p^{(i+1)})+d(p^{(i+1)},\blopt)\right) \nonumber \\
			&= 2\cdot d(p^{(i+1)},\blopt)+d(\o,\blopt)+d(\o,p^{(i+1)})\nonumber
			\\&\geq 3\cdot d(p^{(i+1)},\blopt). \label{eq:third}
		\end{align}
		Combined, we obtain that the algorithm finishes not later than
		\begin{align*}
			t^{(i+1)}+|S^{(i+1)}| \overset{\text{\cref{eq:case222_alg_estimate1}}}&{\leq} \alpha\cdot\opt(t^{(i)}) + |S^{(i)}| + d(p^{(i+1)},\blopt) + \opt(t^{(i+1)}) - T \\[-.1cm]
			\overset{\eqref{eq:third}}&{\leq} \alpha\cdot\opt(t^{(i)}) + |S^{(i)}| + \frac{|S^{(i)}|+T}{3} + \opt(t^{(i+1)}) - T \\[-.1cm]
			\overset{\text{Obs \ref{defobs}a)}}&{\leq} 2 \cdot \opt(t^{(i+1)})+\frac{4}{3} |S^{(i)}|-\frac{2}{3}T\\[-.1cm]
			\overset{\text{Lem \ref{waoa}, \cref{eq:Tlower}}}&{\leq} 2\cdot \opt(t^{(i+1)}) + \frac{2}{3} \opt (t^{(i)})\\[-.1cm]
			\overset{\text{Obs \ref{defobs}a)}}&{\leq} \left( 2+ \frac{2}{3\alpha} \right) \cdot \opt(t^{(i+1)})\\[-.1cm]
			&\leq (1+\alpha) \cdot \opt(t^{(i+1)})
		\end{align*}
		where the last inequality holds if and only if $\alpha \geq \frac{1}{2}+\sqrt{11/12}$.
	\qed\end{proof}
	
	The above results enable us to prove \cref{mainthm}.
	\begin{proof}[of \cref{mainthm}]
		Our goal was to prove by induction that every schedule is $\alpha$-good for $\alpha\geq \frac{1}{2}+\sqrt{11/12}$.
		In \cref{basecase}, we have proven the base case.
		In the induction step, we have distinguished several cases.
		First, we have seen in \cref{interrupted} that the induction step holds if the previous schedule was interrupted.
		Next, we have seen in \cref{waoa} that the induction hypothesis implies $|S^{(i+1)}|\leq\opt(t^{(i+1)})$.
		If the previous schedule was not interrupted, we have first seen in \cref{wrongorder} that the induction step holds if $\opt[t^{(i+1)}]$ loads $r_{f,\opt}^{(i+1)}$ before serving $\rlopt$.
		If $\opt[t^{(i+1)}]$ serves $\rlopt$ before loading $r_{f,\opt}^{(i+1)}$, the induction step holds by \cref{rightorder1} and \cref{rightorder3}.
	\qed\end{proof}

	\section{Analysis on the half-line}
	
	In this section, we prove that $\lazy$ is even better if the metric space considered is the half-line.
	In particular, we prove \cref{half-line_result}.
	To do this, we begin by showing that $\alpha+1$ is an upper bound on the competitive ratio of $\wait(\alpha)$ for $\alpha=\smash{\frac{1+\sqrt{3}}{2}}\approx1.366$.
	Later, we complement this upper bound with a lower bound construction and show that, for all $\alpha\geq0$, $\wait(\alpha)$ has a competitive ratio of at least $\smash{\frac{3+\sqrt{3}}{2}}\approx2.366$.
	
	Since, for all \mbox{$\alpha\geq\frac{1+\sqrt{3}}{2}\thickapprox 1.366$}, Observations~\ref{basecase} and~\ref{interrupted}, as well as Lemmas~\ref{waoa}-\ref{rightorder2} hold, it remains to show a counterpart to \cref{rightorder3} for $\smash{\alpha\geq\frac{1+\sqrt{3}}{2}}$ on the half-line.
	Similarly to the proof of \cref{mainthm}, combining Observations~\ref{basecase} and~\ref{interrupted} and Lemmas~\ref{waoa}-\ref{rightorder2} with \cref{rightorder3_half_line} then yields \cref{half-line_result}.
	
	\begin{lemma}\label{rightorder3_half_line}
		Let $\frac{1+\sqrt{3}}{2} \leq\alpha\leq 2$, and let $M=\R_{\geq0}$.
		 Assume that the $i$-th schedule is $\alpha$-good and was not interrupted, and that $\opt[t^{(i+1)}]$ serves $\rllaz$ before collecting $r_{f,\opt}^{(i+1)}$. If \eqref{eq:dlarge} holds, then $t^{(i+1)}+|S^{(i+1)}|\leq (1+\alpha) \cdot \opt (t^{(i+1)})$.
	\end{lemma}
	\begin{proof}
		First, observe that, in \cref{rightorder3}, we have the same assumptions except that we worked on general metric spaces. Therefore, all the inequalities shown in \cref{rightorder3} hold in this setting, too, so that we can use them for our proof. Next, note that on the half-line, we have for any $x,y \in M$
		\begin{equation}\label{eq:half-line}
		d(x,y)\leq \max \{d(x,\o),d(y,\o)\}.
		\end{equation}
		We show that this implies that a similar claim as in \cref{rightorder3} holds.\\
		\emph{Claim: $\opt[t^{(i+1)}]$ serves all requests in $R^{(i)}$ before picking up~$r_{f,\opt}^{(i+1)}$ in~$\aafopt$.}\\	
		To prove the claim, assume otherwise, i.e., that $\opt[t^{(i+1)}]$ serves $\rlopt$  after collecting $r_{f,\opt}^{(i+1)}$.
		The request~$r_{f,\opt}^{(i+1)}$ is released after schedule~$S^{(i)}$ is started, i.e., after time $\alpha\cdot\opt(t^{(i)})$. Thus,
		\begin{align}\label{eq:case221_opt_estimate1_half-line}
		\opt(t^{(i+1)}) &\geq \alpha\cdot\opt(t^{(i)})+d(\aafopt,\blopt) \nonumber \\
		& \geq \alpha\cdot\opt(t^{(i)})+d(\aafopt,p^{(i+1)}) - d(\blopt, p^{(i+1)}) \nonumber\\
		\overset{\cref{eq:half-line}}&{\geq} \alpha\cdot\opt(t^{(i)})+d(\aafopt,p^{(i+1)}) \nonumber\\
		&\hspace*{1cm}-\max\{d(\blopt,\o),d(\o,p^{(i+1)})\}.
		\end{align}
		Combining the above with the results from the proof of \cref{rightorder3} yields
		\begin{align*}
		\opt(t^{(i+1)})
		\overset{\eqref{eq:case221_opt_estimate1_half-line},\eqref{eq:d(0,x_ell)-estimate}}&{\geq} \alpha\cdot\opt(t^{(i)})+d(\aafopt,p^{(i+1)})\\
		&\hspace{1cm}-\max\Bigl\{\frac{\opt(t^{(i)})+d(p^{(i+1)},\o)}{2},d(\o,p^{(i+1)})\Bigr\} \\
		\overset{\eqref{eq:d(0,x_i+1)-estimate}}&{>} \alpha\cdot\opt(t^{(i)})+d(\aafopt,p^{(i+1)}) \\
		&\hspace*{1cm}-\max\Bigl\{\frac{2\opt(t^{(i)})-(\alpha-1)\opt(t^{(i+1)})}{2},\\
		&\hspace*{4cm}\opt(t^{(i)})-(\alpha-1)\opt(t^{(i+1)})\Bigr\} \\
		& = \frac{\alpha-1}{2}\opt(t^{(i+1)}) + (\alpha-1)\opt(t^{(i)}) + d(\aafopt,p^{(i+1)}) \\
		\overset{\eqref{eq:dlarge}}&{>}  \Bigl(\frac{\alpha-1}{2}+\alpha\Bigr)\opt(t^{(i+1)}) + (\alpha-2)\opt(t^{(i)}) \\
		\overset{\text{Obs \ref{defobs}a)}}&{\geq} \Bigl(\frac{\alpha-1}{2}+\alpha+\frac{\alpha-2}{\alpha}\Bigr)\opt(t^{(i+1)}) \\
		& \geq \opt(t^{(i+1)})
		\end{align*}
		where the last inequality holds for all $\alpha\geq\frac{4}{3}$. As this is a contradiction, we have that $\opt[t^{(i+1)}]$ serves all requests in $R^{(i)}$ before picking up~$r_{f,\opt}^{(i+1)}$ in~$\aafopt$. This completes the proof of the claim.\\
		
		Now that we have established the claim, we turn back to the proof of \cref{rightorder3_half_line}. Let~$T\geq 0$ denote the time it takes $\opt[t^{(i+1)}]$ until it has served $\rlopt$, i.e., all requests from $R^{(i)}$. First, observe that
		\begin{equation}\label{eq:Tlower_hl}
		T \geq \opt(t^{(i)}).
		\end{equation}
		If $p^{(i+1)}\geq\blopt$, as $\opt[t^{(i+1)}]$ visits $p^{(i+1)}$ before $\blopt$, we have
		\begin{equation*}
		T \geq d(\o,p^{(i+1)}) + d(p^{(i+1)},\blopt) \overset{\eqref{eq:half-line}}{\geq} 2\cdot d(p^{(i+1)},\blopt).
		\end{equation*}
		Otherwise, if $p^{(i+1)}<\blopt$, as $S^{(i)}$ visits $\blopt$ before $p^{(i+1)}$, we have
		\begin{equation*}
		T \overset{\eqref{eq:Tlower_hl}}{\geq} \opt(t^{(i)}) \geq |S^{(i)}| \geq d(\o,\blopt) + d(\blopt,p^{(i+1)}) \overset{\eqref{eq:half-line}}{\geq} 2\cdot d(\blopt,p^{(i+1)}).
		\end{equation*}
		Thus, in either case, we have
		\begin{equation}\label{half-line_distance2}
		d(\blopt,p^{(i+1)}) \leq \frac{T}{2}.
		\end{equation}
		Combined, we obtain that the algorithm finishes not later than
		\begin{align*}
		t^{(i+1)}+|S^{(i+1)}| \overset{\text{\cref{eq:case222_alg_estimate1}}}&{\leq} \alpha\cdot\opt(t^{(i)}) + |S^{(i)}| + d(p^{(i+1)},\blopt) + \opt(t^{(i+1)})- T \\
		\overset{\eqref{half-line_distance2}}&{\leq} \alpha\cdot\opt(t^{(i)}) + |S^{(i)}| + \frac{T}{2} + \opt(t^{(i+1)}) - T \\
		\overset{\text{Obs \ref{defobs}a)}}&{\leq} 2 \cdot \opt(t^{(i+1)}) + |S^{(i)}| - \frac{T}{2} \\
		\overset{\text{Lem \ref{waoa}, \cref{eq:Tlower_hl}}}&{\leq} 2\cdot \opt(t^{(i+1)}) + \frac{1}{2} \opt(t^{(i)})\\
		\overset{\text{Obs \ref{defobs}a)}}&{\leq} \left( 2+ \frac{1}{2\alpha} \right) \cdot \opt(t^{(i+1)})\\
		&\leq (1+\alpha) \cdot \opt(t^{(i+1)})
		\end{align*}
		where the last inequality holds if and only if $\alpha \geq \frac{1+\sqrt{3}}{2}\thickapprox 1.366$.
		\qed
	\end{proof}

	\subsection{Lower bound on the half-line}

In this section, we prove Theorem~\ref{half-line_lower-bound}, i.e., we show that $\lazy(\alpha)$ has a competitive ratio of at least $\smash{\frac{3+\sqrt{3}}{2}}\thickapprox 2.366$ for every choice of the parameter $\alpha$, even when the metric space is the half line.
Our proof builds on some lower bound constructions given in \cite{waoa}. 
We begin by restating needed results.

	\begin{lemma}[\cite{waoa}]\label{lem:1plusalpha}
		$\lazy(\alpha)$ has competitive ratio at least $1+\alpha$ for the open online dial-a-ride problem on the half-line for all $\alpha \geq 0$ and every capacity~$c \in \N \cup \{\infty\}$. 
	\end{lemma}
		
The following bound holds, since \cite[Proposition 2]{waoa} only uses the half line.
	
	\begin{lemma}[\cite{waoa}]\label{thm:lower_bound_alpha-leq-1}
		For every $\alpha\in[0,1)$, $\wait (\alpha)$ has a competitive ratio of at least $1+\frac{3}{\alpha+1}>2.5$ for the open online dial-a-ride problem on the half-line for every capacity $c \in \N \cup \{\infty\}$. 
	\end{lemma}
	
	Combining these two results gives that, for $\alpha\in[0,1)\cup[\frac{1+\sqrt{3}}{2},\infty)$, the competitive ratio of $\wait(\alpha)$ is at least $\smash{\frac{3+\sqrt{3}}{2}}\thickapprox 2.366$.
	
The next proposition closes the gap between $\alpha<1$ and $\alpha \geq 1.366$ and completes the proof of Theorem~\ref{half-line_lower-bound}. An overview of the lower bounds for different domains of $\alpha$ can be found in Figure~\ref{fig:plot_lower_bounds}.
	\begin{proposition}\label{thm:lower_bound_alpha-geq-1}
		For every $\alpha\in [1,1.366)$, the algorithm $\wait (\alpha)$ has a competitive ratio of at least $2+\frac{1}{2\alpha}$ for the open online dial-a-ride problem on the half-line for every capacity $c \in \N \cup \{\infty\}$.
	\end{proposition}
	
	\begin{figure}[t]
	\begin{center}
		\begin{tikzpicture}[scale=0.77,font=\small,decoration=brace,every node/.style={scale=1}]
			\def\endpt{1pt}
			\def\reqpt{2.9pt}
			
			\def\opt{\textsc{Opt}}
			
			\def\xmax{13}
			\def\ymin{0}
			\def\ymax{4}
			
			\def\alph{1.183}
			\def\eps{.3}
			
			\draw[line width=1pt,black,->,>=stealth] (0,{\ymin})--(0,{\ymax});
			\draw[line width=1pt,black,->,>=stealth] ({-0.17},0)--({\xmax},0);
			\node[above left] at (-0.05,-0.2) {$0$};
			\node[below right, align=left] at (0,{\ymax}) {position};
			\node[above left, align=left] at ({\xmax},0) {time};
			
			\draw [fill=ForestGreen] ({\xmax-1},{\ymax-0.15}) rectangle ({\xmax-1.2},{\ymax-0.35});
			\node[below right,align=left] at ({\xmax-1},{\ymax}) {$\wait$};
			\draw [fill=NavyBlue] ({\xmax-1},{\ymax-0.55}) rectangle ({\xmax-1.2},{\ymax-0.75});
			\node[below right,align=left,] at ({\xmax-1},{\ymax-0.4}) {$\opt$};§
			
			
			\draw[line width=1pt,black,-] (0,-0.095)--(0,0);
			\node[align=center] at (0,-0.3) {$0$};
			
			\draw[line width=1pt,black,-] (\eps,-0.095)--(\eps,0);
			\node[align=center] at (\eps,-0.3) {$0$};
			
			\draw[line width=1pt,black,-] (2*\eps,-0.095)--(2*\eps,0);
			\node[align=center] at (2*\eps,-0.3) {$0$};
			
			\draw[line width=1pt,black,-] (4*\alph,-0.095)--(4*\alph,0);
			\node[align=center] at (4*\alph,-0.3) {$\opt=4\alpha$};
			
			\draw[line width=1pt,black,-] ({8*\alph+2-\eps*(2*\alph+2)},-0.095)--({8*\alph+2-\eps*(2*\alph+2)},0);
			\node[align=center] at ({8*\alph+2-\eps*(2*\alph+2)},-0.3) {$\wait$};

			\draw[line width=1pt,black,-] (-0.1,4*\alph-2)--(0,4*\alph-2);
			\node[left,align=right] at (-0.075,4*\alph-2) {$4\alpha-2$};
			
			\draw[line width=1pt,black,-] (-0.1,2-\eps)--(0,2-\eps);
			\node[left,align=right] at (-0.075,{2-\eps}) {$2-\epsilon$};
			
			\draw[line width=1pt,black,-] (-0.075,1)--++(0.075,0);
			\node[left,align=right] at (-0.075,1) {$1$};

			\filldraw[fill=black] (0,0) circle (\reqpt);
			
			\draw[line width=2.5pt,black,-] (0,0)--({\alph*(4-2*\eps)},0)--({\alph*(4-2*\eps)+2-\eps},2-\eps);
			\draw[line width=1.5pt,ForestGreen,-] (0,0)--({\alph*(4-2*\eps)},0)--({\alph*(4-2*\eps)+2-\eps},2-\eps);
			
			\draw[line width=5pt,black,-] ({\alph*(4-2*\eps)},0)--({\alph*(4-2*\eps)+2-\eps},2-\eps);
			\draw[line width=4pt,VioletRed,-] ({\alph*(4-2*\eps)},0)--({\alph*(4-2*\eps)+1},1);
			\draw[line width=2.5pt,black,-] ({\alph*(4-2*\eps)},0)--({\alph*(4-2*\eps)+1},1);
			\draw[line width=1.5pt,ForestGreen,-] ({\alph*(4-2*\eps)},0)--({\alph*(4-2*\eps)+1},1);
			
			\draw[line width=4pt,Orange,-] ({\alph*(4-2*\eps)+1},1)--({\alph*(4-2*\eps)+2-\eps},2-\eps);
			\draw[line width=2.5pt,black,-] ({\alph*(4-2*\eps)+1},1)--({\alph*(4-2*\eps)+2-\eps},2-\eps);
			\draw[line width=1.5pt,ForestGreen,-] ({\alph*(4-2*\eps)+1},1)--({\alph*(4-2*\eps)+2-\eps},2-\eps);
			
			\draw[line width=2.5pt,black,-] ({\alph*(4-2*\eps)+2-\eps},2-\eps)--(({\alph*(4-2*\eps)+3-2*\eps},1);
			\draw[line width=5pt,black,-] ({\alph*(4-2*\eps)+3-2*\eps},1)--({\alph*(4-2*\eps)+4-2*\eps},0);
			\draw[line width=4pt,Goldenrod,-] ({\alph*(4-2*\eps)+3-2*\eps},1)--({\alph*(4-2*\eps)+4-2*\eps},0);
			\draw[line width=2.5pt,black,-] ({\alph*(4-2*\eps)+3-2*\eps},1)--({\alph*(4-2*\eps)+4-2*\eps},0);
			\draw[line width=1.5pt,ForestGreen,-] ({\alph*(4-2*\eps)+2-\eps},2-\eps)--({\alph*(4-2*\eps)+4-2*\eps},0);
			
			\draw[line width=2.5pt,black,-] ({\alph*(4-2*\eps)+4-2*\eps},0)--({8*\alph+2-\eps*(2*\alph+2)},4*\alph-2);
			\draw[line width=1.5pt,ForestGreen,-] ({\alph*(4-2*\eps)+4-2*\eps},0)--({8*\alph+2-\eps*(2*\alph+2)},4*\alph-2);
			
			\filldraw[fill=VioletRed] 	({\alph*(4-2*\eps)},0) circle (\reqpt);
			\filldraw[fill=Orange] 	({\alph*(4-2*\eps)+1},1) circle (\reqpt);
			\draw[fill=VioletRed,VioletRed] (4.97,1.06) arc (130:310:{\reqpt-.5pt});
			\draw[line width=0.4pt,black,-] (4.96,1.07)--({5.1},0.94);
			
			\filldraw[fill=Orange] 	({\alph*(4-2*\eps)+2-\eps},2-\eps) circle (\reqpt);

			\filldraw[fill=Goldenrod] 	({\alph*(4-2*\eps)+3-2*\eps},1) circle (\reqpt);
			\filldraw[fill=Goldenrod] 	({\alph*(4-2*\eps)+4-2*\eps},0) circle (\reqpt);
			\filldraw[fill=Purple] 	({8*\alph+2-\eps*(2*\alph+2)},4*\alph-2) circle (\reqpt);
			
			\draw[line width=5pt,black,-] (0,0)--(1,1);
			\draw[line width=4pt,VioletRed,-] (0,0)--(1,1);
			\draw[line width=2.5pt,black,-] (0,0)--(1,1);
			\draw[line width=1.5pt,NavyBlue,-] (0,0)--(1,1);

			\draw[line width=5pt,black,-] (1,1)--(2,0);
			\draw[line width=4pt,Goldenrod,-] (1,1)--({2},0);
			\draw[line width=2.5pt,black,-] (1,1)--(2,0);
			\draw[line width=1.5pt,NavyBlue,-] (1,1)--(2,0);
			
			\draw[line width=5pt,black,-] (3,1)--(4-\eps,2-\eps);
			\draw[line width=4pt,Orange,-] (3,1)--(4-\eps,2-\eps);
			\draw[line width=2.5pt,black,-] (2,0)--(4-\eps,2-\eps);
			\draw[line width=1.5pt,NavyBlue,-] (2,0)--(4-\eps,2-\eps);
			\draw[line width=2.7pt,black,-] (4-\eps,2-\eps)--(4*\alph,4*\alph-2);
			\draw[line width=2pt,NavyBlue,-] (4-\eps,2-\eps)--(4*\alph,4*\alph-2);
			
			\filldraw[fill=VioletRed] 	({1},1) circle (\reqpt);
			\draw[fill=Goldenrod,Goldenrod] (1,.92) arc (270:450:{\reqpt-.5pt});
			\draw[line width=0.4pt,black,-] (1,0.9)--(1,1.1);
			\filldraw[fill=Goldenrod] 	({2},0) circle (\reqpt);
			\filldraw[fill=Orange] 	(3,1) circle (\reqpt);
			\filldraw[fill=Orange] 	(4-\eps,2-\eps) circle (\reqpt);

			\draw[line width=2.5pt,black,densely shadow,->,>=stealth] (0,0)--(0,1);
			\filldraw[fill=VioletRed] (0,0) circle (\reqpt);
			\draw[line width=1.5pt,VioletRed,densely dashed,->,>=stealth] (0,0)--(0,1);
			
			\draw[line width=2.5pt,black,densely shadow,->,>=stealth] (\eps,1)--(\eps,2-\eps);
			\filldraw[fill=Orange] (\eps,1) circle (\reqpt);
			\draw[line width=1.5pt,Orange,densely dashed,->,>=stealth] (\eps,1)--(\eps,2-\eps);
			
			\draw[line width=2.5pt,black,densely shadow,->,>=stealth] (2*\eps,1)--(2*\eps,0);
			\filldraw[fill=Goldenrod] (2*\eps,1) circle (\reqpt);
			\draw[line width=1.5pt,Goldenrod,densely dashed,->,>=stealth] (2*\eps,1)--(2*\eps,0);
			
			\filldraw[fill=Purple] (4*\alph,4*\alph-2) circle (\reqpt);
		\end{tikzpicture}
	\end{center}
	\caption{Instance of the open online dial-a-ride problem on the half-line where $\wait(\alpha)$ has a competitive ratio of at least $2+\frac{1}{2\alpha}$ for all $\alpha\in[1,1.366)$.}
	\label{fig:lower_bound_big_alpha}
\end{figure}
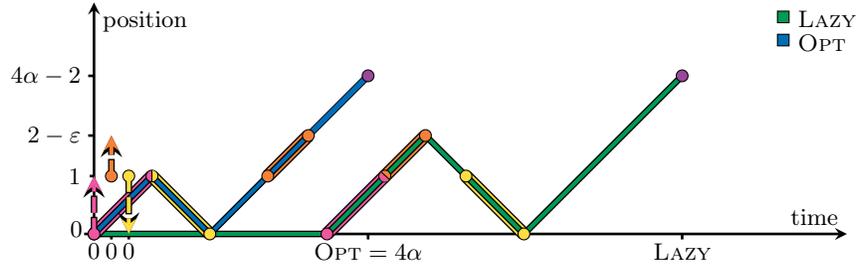
	
\begin{proof}
	Let $\alpha\in [1,1.366)$ and let $\epsilon>0$ be sufficiently small.
	We define an instance (cf. Figure~\ref{fig:lower_bound_big_alpha}) by giving the request sequence
	\begin{align*}
	\quad \quad \quad \quad \quad \quad \quad \quad 
	&r_1=\left(0,1;0\right) \textrm{, } r_2=(1,0;0) \textrm{, } r_3=(1,2-\epsilon;0) \textrm{,}\\
	&\textrm{and } r_4=(4\alpha-2,4\alpha-2;4\alpha).
	\end{align*}
	The offline optimum delivers the requests in the order $(r_1,r_2,r_3,r_4)$ with no waiting times.
	This takes $4\alpha$ time units.
	
	On the other hand because $\opt(0)=4-2\epsilon$, $\wait(\alpha)$ waits in the origin until time $\alpha(4-2\epsilon)$ and starts serving requests $r_1,r_2,r_3$ in the order $(r_1,r_3,r_2)$.
	At time $4\alpha$, request $r_4$ is released.
	At this time, serving the loaded request $r_1$ and returning to the origin takes time
	\begin{equation*}
	\alpha(4-2\epsilon)+2 
	=4\alpha+2 -2\alpha \epsilon
	\overset{\alpha \in [1,1.366), \epsilon \ll 1}{>}4\alpha^2=\alpha\cdot\opt(4\alpha).
	\end{equation*}	 
	Thus, $\wait(\alpha)$ continues its schedule and afterwards serves $r_4$.
	Overall, this takes time (at least)
	$\alpha(4-2\epsilon)+(4-2\epsilon)+4\alpha-2=8\alpha +2 - (2\alpha+2) \epsilon$.
	Letting $\epsilon\to 0$, we obtain that the competitive ratio of $\wait(\alpha)$ is at least
	\begin{align*}
	\lim_{\varepsilon \to 0}\frac{8\alpha +2 - (2\alpha+2) \epsilon}{4\alpha}= 2+\frac{1}{2\alpha}. 
	\end{align*}\qed
\end{proof}

\begin{figure}[t]
	\vspace{.5cm}
	\begin{center}
		\begin{tikzpicture}[scale=0.85]
		\begin{axis}[
		samples=100,
		xmin=0,
		xmax=3,
		xtick={0,1,2,...,4},
		ytick={2,3,4},
		ymin=2,
		ymax=4,
		xlabel=waiting parameter $\alpha$,
		ylabel=competitive ratio $\rho$,
		height=8cm,
		width=12cm,
		ylabel near ticks,
		xlabel near ticks,
		]
		
		\addplot[no markers, thick, domain=1:1.366][black!40!red,line width=1pt, dashed]{x+1};
		\addplot[no markers, thick, domain=1.366:3][black!40!red,line width=1pt]{x+1};
		
		\addplot[no markers, domain=1:1.366][name path=B2,black!40!green,line width=1pt]{2+1/2/x};
		
		\addplot[no markers, domain=0:1][name path=B1,black!40!blue,line width=1pt]{1+(3)/(1+x)};
		
		
		\draw[densely dashed,line width=0.5pt,black!70] (axis cs:1.366,\pgfkeysvalueof{/pgfplots/ymin}) -- (axis cs:1.366,\pgfkeysvalueof{/pgfplots/ymax});
		\draw[densely dashed,line width=0.5pt,black!70] (axis cs:1,\pgfkeysvalueof{/pgfplots/ymin}) -- (axis cs:1,\pgfkeysvalueof{/pgfplots/ymax});
		\addplot[no markers,densely dashed, domain=0:3,line width=0.5pt, black!70]{2.366};
		\addplot[no markers,densely dotted, domain=0:3,line width=0.5pt, black]{2.366};
		
		\addplot[no markers, domain=0:3][name path=A,black,line width=.5pt]{4};
		
		\addplot[black!6!red!15, postaction={pattern=my north east lines,pattern color=gray}] fill between[of=A and B1];
		\addplot[black!6!red!15, postaction={pattern=my north east lines,pattern color=gray}] fill between[of=A and B2];
		
		\node[above] at (axis cs: 1.6,2.05) {$\alpha=1.366$};
		\node[below right] at (axis cs: 2.5,2.52) {$\rho= 2.366$};
		\end{axis}
		\end{tikzpicture}
	\end{center}
	\caption{Lower bounds on the competitive ratio of $\wait(\alpha)$ depending on~$\alpha$. The lower bound of \cref{lem:1plusalpha} is depicted in red, the lower bound of \cref{thm:lower_bound_alpha-leq-1} in blue, and the lower bound of \cref{thm:lower_bound_alpha-geq-1} in green.
	}
	\label{fig:plot_lower_bounds}
\end{figure}
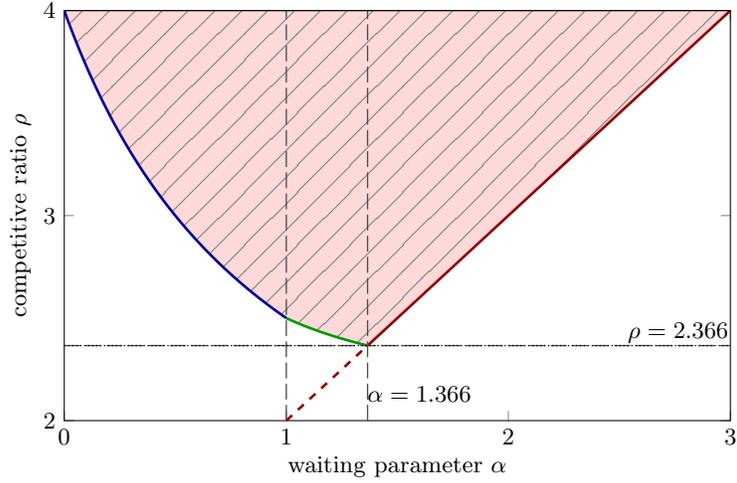

%

\appendix

\section{Factor-revealing approach for the half-line}\label{sec:fr_hl}

We show how to use the factor~revealing approach from Section~\ref{sec:factorRevealing} for the dial-a-ride problem on the half-line.
Consider the following variables (recall that $k\in\N$ is the number of schedules started by $\lazy(\alpha)$).
\begin{itemize}
	\item $t_1 = t^{(k-1)}$, the start time of the second to last schedule
	\item $t_2 = r^{(k)}$, the start time of the last schedule
	\item $s_1 = |S^{(k-1)}|$, the duration of the second to last schedule
	\item $s_2 = |S^{(k)}|$, the duration of the last schedule
	\item $\opt_1 = \opt(t^{(k-1)})$, duration of the optimal tour serving requests released until $t^{(k-1)}$
	\item $\opt_2 = \opt(t^{(k)})$, duration of the optimal tour
	\item $p_1 = p^{(k)}$, the position where $\lazy(\alpha)$ ends the second to last schedule
	\item $p_2 = a^{(k)}_{f,\opt}$, the position of the first request in $R^{(k)}$ picked up first by the optimal tour
	\item $s_2^a = |S(R^{(k)},a^{(k)}_{f,\opt})|$, duration of the schedule serving $R^{(k)}$ starting in $p_2$
	\item $d = d(p^{(k)},a^{(k)}_{f,\opt})$, the distance between $p_1$ and $p_2$
\end{itemize}
With these variables
\[
x=\bigl(t_1,t_2,s_1,s_2,\opt_1,\opt_2,p_1,p_2,s_2^a,d\bigr),
\]
we can create the following valid optimization problem.
\begin{alignat}{4}
\max && \quad t_2 + s_2\nonumber \\
\text{s.t.} && \opt_2 = &\ 1\label{factorRevealingLP_1} \\
&& d = &\ |p_1 - p_2|\label{factorRevealingLP_2} \\
&& t_2 = &\ \max\{t_1+s_1,\alpha\opt_2\}\label{factorRevealingLP_3} \\
&& t_1 \geq &\ \alpha\opt_1\label{factorRevealingLP_4} \\
&& \opt_1 \geq &\ p_1\label{factorRevealingLP_5} \\
&& s_2 \leq &\ d + s_2^a\label{factorRevealingLP_6} \\
&& \opt_2 \geq &\ t_1 + s_2^a\label{factorRevealingLP_7} \\
&& t_1 + s_1 \leq &\ (1+\alpha)\opt_1\label{factorRevealingLP_8} \\
&& \opt_2 \geq &\ p_1 + d \quad\quad\quad\quad\quad\quad \text{or} \quad\quad \opt_2 \geq t_1 + d\label{factorRevealingLP_9} \\
&& d \geq &\ \alpha\opt_2 - \opt_1 \quad\quad \text{or} \quad\quad s_1 - p_1 \leq 2(\opt_2 - p_2)\label{factorRevealingLP_10} \\
&& x \geq &\ 0
\end{alignat}
Note that in \eqref{factorRevealingLP_9} and \eqref{factorRevealingLP_10}, at least one of the two inequalities has to be satisfied in each case.
In order to obtain an MILP, one can introduce four binary variables $b_1,\dots,b_4$ to model constraints~\eqref{factorRevealingLP_2},~\eqref{factorRevealingLP_3},~\eqref{factorRevealingLP_9}, and~\eqref{factorRevealingLP_10}.

With $M>0$ being a large enough constant, equality \eqref{factorRevealingLP_2} can be replaced by the inequalities
\begin{align*}
d \geq p_1 - p_2, \\
d \geq p_2 - p_1, \\
d \leq p_1 - p_2 + b_1\cdot M, \\
d \leq p_2 - p_1 + (1-b_1)\cdot M.
\end{align*}
Equality \eqref{factorRevealingLP_3} can be replaced by the inequalities
\begin{align*}
t_2 \geq t_1 + s_1, \\
t_2 \geq \alpha\opt_2, \\
t_2 \leq t_1 + s_1 + b_2\cdot M, \\
t_2 \leq \alpha\opt_2 + (1-b_2)\cdot M.
\end{align*}
Constraint \eqref{factorRevealingLP_9} can be replaced by the inequalities
\begin{align*}
\opt_2 \geq p_1 + d - b_3\cdot M, \\
\opt_2 \geq t_1 + d - (1-b_3)\cdot M,
\end{align*}
and, likewise, \eqref{factorRevealingLP_10} by the inequalities
\begin{align*}
d \geq \alpha\opt_2 - \opt_1 - b_4\cdot M, \\
s_1 - p_1 \leq 2(\opt_2 - p_2) + (1-b_4)\cdot M.
\end{align*}

The resulting MILP has the optimal solution
\begin{align*}
&\ \bigl(t_1,t_2,s_1,s_2,\opt_1,\opt_2,p_1,p_2,s_2^a,d,b_1,b_2,b_3,b_4\bigr) \\
=&\ \Bigl(1,\frac{\alpha+1}{\alpha},\frac{1}{\alpha},2-\alpha,\frac{1}{\alpha},1,0,2-\alpha,0,2-\alpha,0,1,1,1\Bigr)
\end{align*}
and optimal value $\max\{3+\frac{1}{\alpha}-\alpha,1+\alpha\}$.
For $\alpha=\frac{1+\sqrt{3}}{2}>1.366$, this expression is minimized.

	%
	%
	%
	
	\bibliographystyle{splncs04}
	\bibliography{Lazy}

\begin{thebibliography}{10}
\providecommand{\url}[1]{\texttt{#1}}
\providecommand{\urlprefix}{URL }
\providecommand{\doi}[1]{https://doi.org/#1}

\bibitem{AscheuerKrumkeRambau/00}
Ascheuer, N., Krumke, S.O., Rambau, J.: Online dial-a-ride problems: Minimizing
  the completion time. In: Proceedings of the 17th Annual Symposium on
  Theoretical Aspects of Computer Science (STACS). pp. 639--650 (2000)

\bibitem{AusielloDemangeLauraPaschos/04}
Ausiello, G., Demange, M., Laura, L., Paschos, V.: Algorithms for the on-line
  quota traveling salesman problem. Information Processing Letters
  \textbf{92}(2),  89--94 (2004)

\bibitem{Ausiello/01}
Ausiello, G., Feuerstein, E., Leonardi, S., Stougie, L., Talamo, M.: Algorithms
  for the on-line travelling salesman. Algorithmica  \textbf{29}(4),  560--581
  (2001)

\bibitem{AusielloAllulliBonifaciLaura/06}
Ausiello, G., Allulli, L., Bonifaci, V., Laura, L.: On-line algorithms, real
  time, the virtue of laziness, and the power of clairvoyance. In: Proceddings
  of the 3rd International Conference on Theory and Applications of Models of
  Computation (TAMC). pp. 1--20 (2006)

\bibitem{waoa}
Baligács, J., Disser, Y., Mosis, N., Weckbecker, D.: An improved algorithm for
  open online dial-a-ride. In: Proceedings of the 20th Workshop on
  Approximation and Online Algorithms (WAOA) (2022)

\bibitem{BienkowskiKL/21}
Bienkowski, M., Kraska, A., Liu, H.: Traveling repairperson, unrelated
  machines, and other stories about average completion times. In: Proceedings
  of the 48th International Colloquium on Automata, Languages, and Programming
  (ICALP). pp. 28:1--28:20 (2021)

\bibitem{BienkowskiLiu/19}
Bienkowski, M., Liu, H.: An improved online algorithm for the traveling
  repairperson problem on a line. In: Proceedings of the 44th International
  Symposium on Mathematical Foundations of Computer Science (MFCS). pp.
  6:1--6:12 (2019)

\bibitem{Birx/20}
Birx, A.: Competitive analysis of the online dial-a-ride problem. Ph.D. thesis,
  TU Darmstadt (2020)

\bibitem{BirxDisser/20}
Birx, A., Disser, Y.: Tight analysis of the smartstart algorithm for online
  dial-a-ride on the line. SIAM Journal on Discrete Mathematics
  \textbf{34}(2),  1409--1443 (2020)

\bibitem{BirxDisser/22}
Birx, A., Disser, Y., Schewior, K.: Improved bounds for open online dial-a-ride
  on the line. Algorithmica  (2022)

\bibitem{BjeldeDisserHackfeldEtal/20}
Bjelde, A., Disser, Y., Hackfeld, J., Hansknecht, C., Lipmann, M., Mei{\ss}ner,
  J., Schewior, K., Schl\"oter, M., Stougie, L.: Tight bounds for online tsp on
  the line. ACM Transactions on Algorithms  \textbf{17}(1) (2020)

\bibitem{BlomKrumkePaepeStougie/01}
Blom, M., Krumke, S.O., de~Paepe, W.E., Stougie, L.: The online {TSP} against
  fair adversaries. INFORMS Journal on Computing  \textbf{13}(2),  138--148
  (2001)

\bibitem{BonifaciStougie/08}
Bonifaci, V., Stougie, L.: Online $k$-server routing problems. Theory of
  Computing Systems  \textbf{45}(3),  470--485 (2008)

\bibitem{FeuersteinStougie/01}
Feuerstein, E., Stougie, L.: On-line single-server dial-a-ride problems.
  Theoretical Computer Science  \textbf{268}(1),  91--105 (2001)

\bibitem{HauptmeierKrumkeRambauWirth/01}
Hauptmeier, D., Krumke, S., Rambau, J., Wirth, H.C.: Euler is standing in line
  dial-a-ride problems with precedence-constraints. Discrete Applied
  Mathematics  \textbf{113}(1),  87--107 (2001)

\bibitem{HauptmeierKrumkeRabau/00}
Hauptmeier, D., Krumke, S.O., Rambau, J.: The online dial-a-ride problem under
  reasonable load. In: Proceedings of the 4th Italian Conference on Algorithms
  and Complexity (CIAC). pp. 125--136 (2000)

\bibitem{JailletLu/11}
Jaillet, P., Lu, X.: Online traveling salesman problems with service
  flexibility. Networks  \textbf{58}(2),  137--146 (2011)

\bibitem{JailletLu/14}
Jaillet, P., Lu, X.: Online traveling salesman problems with rejection options.
  Networks  \textbf{64}(2),  84--95 (2014)

\bibitem{JailletWagner/08}
Jaillet, P., Wagner, M.R.: Generalized online routing: New competitive ratios,
  resource augmentation, and asymptotic analyses. Operations Research
  \textbf{56}(3),  745--757 (2008)

\bibitem{JawgalMuralidharaSrinivasan/19}
Jawgal, V.A., Muralidhara, V.N., Srinivasan, P.S.: Online travelling salesman
  problem on a circle. In: In Proceedings of the 15th International Conference
  on Theory and Applications of Models of Computation (TAMC). pp. 325--336
  (2019)

\bibitem{Krumke0}
Krumke, S.O.: Online optimization competitive analysis and beyond. Habilitation
  thesis, Zuse Institute Berlin (2001)

\bibitem{Krumke/02}
Krumke, S.O., Laura, L., Lipmann, M., Marchetti-Spaccamela, A., de~Paepe, W.,
  Poensgen, D., Stougie, L.: Non-abusiveness helps: {A}n {O}(1)-competitive
  algorithm for minimizing the maximum flow time in the online traveling
  salesman problem. In: Proceedings of the 5th International Workshop on
  Approximation Algorithms for Combinatorial Optimization (APPROX). pp.
  200--214 (2002)

\bibitem{KrumkePaepePoensgenEtal/06}
Krumke, S.O., de~Paepe, W.E., Poensgen, D., Lipmann, M., Marchetti-Spaccamela,
  A., Stougie, L.: On minimizing the maximum flow time in the online
  dial-a-ride problem. In: Proceedings of the 3rd International Conference on
  Approximation and Online Algorithms (WAOA). pp. 258--269 (2006)

\bibitem{KrumkePaepePoensgenStougie/03}
Krumke, S.O., de~Paepe, W.E., Poensgen, D., Stougie, L.: News from the online
  traveling repairman. Theoretical Computer Science  \textbf{295}(1-3),
  279--294 (2003)

\bibitem{Lipmann/03}
Lipmann, M.: On-line routing. Ph.D. thesis, Technische Universiteit Eindhoven
  (2003)

\bibitem{Lipmann/04}
Lipmann, M., Lu, X., de~Paepe, W.E., Sitters, R.A., Stougie, L.: On-line
  dial-a-ride problems under a restricted information model. Algorithmica
  \textbf{40}(4),  319--329 (2004)

\end{thebibliography}
\end{document}